\documentclass[11pt,a4paper]{article}
\usepackage[utf8]{inputenc}
\usepackage[T1]{fontenc}
\usepackage{amsmath,amssymb,amsthm}
\usepackage{mathtools}
\usepackage[margin=2.6cm]{geometry}
\usepackage{enumitem}
\usepackage{needspace}
\usepackage{xcolor}
\setlist{topsep=3pt,itemsep=1.5pt,parsep=0pt,partopsep=0pt}
\usepackage{hyperref}
\hypersetup{
  colorlinks=true,linkcolor=blue!50!black,citecolor=blue!50!black,urlcolor=blue!50!black,
  pdftitle={The Complexity of Finding Additional Solutions},
  pdfauthor={Saman Farhat},
  pdfsubject={Fixed supplied lists, arbitrary supplied lists, and additional-solution search degrees},
  pdfkeywords={additional-solution problem, fixed supplied list, arbitrary supplied list,
    FNP, one-query reduction, search degree}
}

\theoremstyle{plain}
\newtheorem{theorem}{Theorem}[section]
\newtheorem{lemma}[theorem]{Lemma}
\newtheorem{proposition}[theorem]{Proposition}
\newtheorem{corollary}[theorem]{Corollary}
\theoremstyle{definition}
\newtheorem{definition}[theorem]{Definition}
\newtheorem{convention}[theorem]{Convention}
\newtheorem{remark}[theorem]{Remark}
\newtheorem{observation}[theorem]{Observation}

\newcommand{\Sig}{\Sigma}
\newcommand{\ran}{\operatorname{ran}}
\newcommand{\dm}{\operatorname{dom}}
\newcommand{\dg}{\operatorname{deg}}
\newcommand{\FNP}{\mathrm{FNP}}
\newcommand{\TFNP}{\mathrm{TFNP}}
\newcommand{\NP}{\mathrm{NP}}
\newcommand{\PP}{\mathrm{P}}
\newcommand{\lmT}{\leq^p_{1\text{-}T}}
\newcommand{\emT}{\equiv^p_{1\text{-}T}}
\newcommand{\lm}{\leq^p_m}
\newcommand{\eqm}{\equiv^p_m}
\newcommand{\bl}[1]{\overline{#1}}
\newcommand{\tup}[1]{\langle #1\rangle}
\newcommand{\Val}{\mathrm{Val}}

\newcommand{\Res}{\operatorname{Res}}
\newcommand{\sm}{\leq^p_{\rm sm}}

\title{The Complexity of Finding Additional Solutions}
\author{Saman Farhat\\[4pt]
  {\normalsize The Graduate Center of the City University of New York}\\
  {\normalsize Amirkabir University of Technology}\\[2pt]
  {\small\texttt{saman.farhat33@gc.cuny.edu}\qquad\texttt{saman.farhat@aut.ac.ir}}}
\date{}

\begin{document}
\maketitle

\begin{abstract}
We study the difficulty of finding a new solution when a list of distinct solutions is already supplied. For any fixed finite collection of list lengths, we can prescribe the additional-solution degrees while preserving a given $1$-paddable domain and first-solution degree. We also prove an exact compatibility condition for prescribing an arbitrary positive-list degree at the same time: it must dominate every prescribed fixed-level degree, and this condition suffices. All degrees use polynomial-time reductions with at most one oracle query and correct behavior on empty fibers; the realization needs no assumption about P versus NP. The combined construction gives each underlying witness exponentially many short counter copies. A list that exhausts these copies is long enough to pay for exhaustive search through the underlying witness bound. Separately, we construct a total problem whose first-solution task and every fixed-length additional-solution task are polynomial-time solvable, while its arbitrary-list task is FNP-complete. Assuming P differs from NP, the domain, first degree, and entire fixed-level degree sequence therefore do not determine the arbitrary-list degree. A constant fiber cap restores determination by finitely many fixed levels. The constructions change witness sets and do not prescribe an infinite sequence of degrees or enumerate all remaining solutions.
\end{abstract}

\section{Introduction}\label{sec:intro}

Finding a solution is often only the beginning. A person choosing a schedule may want several alternatives. Someone completing a puzzle may want to know whether another completion exists. In either case, the next task comes with information that the original task did not have: some solutions are already known. We ask how this information affects the difficulty of finding a new solution, and what changes when the number of known solutions grows.

Earlier work by Egidy, Gla{\ss}er, and Godau~\cite{EGG25} showed that finding a first solution and finding a second can have independently chosen difficulties. Starting from a search problem with a $1$-paddable domain, they construct another with exactly the same solvable instances and an equivalent first-solution task, but with a prescribed second-solution difficulty. They asked whether analogous results hold when a fixed number of distinct solutions are supplied, when the supplied list has arbitrary length, or when the goal is to find all remaining solutions~\cite[Section~5]{EGG25}. This paper extends the fixed-list construction to simultaneous finite prescription and determines when a prescribed arbitrary positive-list degree is compatible with those finitely many degrees. It does not address finding all remaining solutions.

Our main result allows several additional-solution tasks to be prescribed within one problem. We can choose the difficulty of finding a second solution, choose the difficulty of finding a third separately, and continue for any fixed finite number of steps. Throughout, the constructed problem has solutions on exactly the same instances as the original problem, and finding a first solution retains its original difficulty. The construction requires a recognizable collection of positive instances in which arbitrary inputs can be encoded and recovered efficiently. Under this assumption, controlling one of these additional-solution tasks places no restriction on the other finitely many tasks being prescribed.

A list of arbitrary length raises a different issue. Every fixed-list task reduces to the positive arbitrary-list task by using the same supplied list. Thus a prescribed arbitrary-list degree must dominate each prescribed fixed-level degree. We prove that this necessary condition is sufficient under the same padding assumption, while retaining the domain and first-solution degree. The new construction gives each underlying witness exponentially many copies identified by short binary counters. If a supplied list exhausts a represented witness's copies, its full encoded length is large enough to pay for searching all candidate underlying witnesses. This is a separate construction from the finite-copy construction, whose arbitrary-list degree we also characterize exactly in terms of the original and target residual tasks.

Assuming $\PP\ne\NP$, we also show that the difficulty for arbitrary-length lists is not determined by the domain, first-solution degree, and difficulties for all fixed list lengths. We construct a total search problem for which finding a first solution and finding another after any fixed number of supplied solutions are polynomial-time solvable. Yet finding another solution after an arbitrary-length supplied list is $\FNP$-complete, and hence is not polynomial-time solvable unless $\mathrm{P}=\mathrm{NP}$. We also clarify when an additional-solution task is guaranteed to have an answer and extend the total-search construction from the earlier work to any one fixed number of supplied solutions.

The contribution is a description of how much freedom there is in the difficulty of finding further solutions, together with an explanation of why fixed and growing lists behave differently. This freedom comes from constructing new solution sets while retaining the original positive instances and first-solution difficulty. The combined result prescribes finitely many fixed levels together with a compatible arbitrary positive-list degree; it does not prescribe an entire infinite sequence. Preserving an unchanged natural notion of solution and computing all remaining solutions remain outside the results established here.

\section{Search relations, supplied lists, and reductions}\label{sec:prelim}

The aim of this section is to establish a common language for comparing
the difficulty of finding a first solution with that of finding an
additional solution. We describe a search problem by specifying the
possible solutions for each instance. An ordinary algorithm must find one of these
solutions, while an additional-solution algorithm must find one outside
a supplied list. We introduce and explain the reductions used for comparing the difficulty of
such tasks, and distinguish lists of fixed length from lists whose
length is part of the input. Finally, we introduce the tools used in
our constructions.

An additional-solution problem needs three conventions: what counts as a
witness, how previously supplied witnesses are represented, and what an
algorithm must do when no answer exists. We fix all three on binary strings.

\subsection{Instances, witnesses, and verification}\label{sec:verification}

Let $\Sig=\{0,1\}$. The set $\Sig^*$ contains all finite binary strings,
and $|x|$ is the number of bits in $x$. The empty string is $\varepsilon$;
it is a possible witness, not an empty set of witnesses. A \emph{search
relation} is a set $f\subseteq\Sig^*\times\Sig^*$. The first coordinate
$x$ is an \emph{instance}, and a second coordinate $y$ with $(x,y)\in f$
is a \emph{witness}, or solution, for that instance. The \emph{fiber}
$f(x)=\{y:(x,y)\in f\}$ is its entire set of witnesses. The \emph{domain}
$\dm(f)=\{x:f(x)\ne\emptyset\}$ consists of the instances that have a
witness. A relation is \emph{total} if its domain is all of $\Sig^*$.

The \emph{graph} of $f$ is the set of pairs $(x,y)$ belonging to the
relation. A \emph{graph verifier} is an algorithm $V_f(x,y)$ that accepts
exactly those pairs. Thus a ``graph test'' means checking a candidate
witness. A polynomial-time graph
verifier halts on every pair within $t_f(|x|+|y|+1)$ steps for some
polynomial $t_f$. This bound includes rejected candidates, however long.
Checking a candidate and finding a candidate are different tasks: an
instance can have many efficiently checkable witnesses without an efficient
algorithm for selecting one.

A \emph{polynomial witness bound}, also called a length bound, is a
nonnegative integer-valued polynomial $b_f$ such that
\[
 y\in f(x)\quad\Longrightarrow\quad |y|\le b_f(|x|).
\]
It bounds every valid witness using the instance length alone. In contrast,
the verifier's time depends on both instance and candidate length. We use
monotone polynomial majorants: replacing a bound by a larger polynomial
with nonnegative coefficients preserves it and ensures that substituting a
larger argument cannot decrease the bound. In particular, if an encoded
input has length $L\ge |x|$, then $b_f(|x|)\le b_f(L)$.

Following the relation model of~\cite[Definitions 2.3--2.8]{EGG25},
$\FNP$ consists of the search relations with a polynomial-time graph
verifier and a polynomial witness bound. Here $\FNP$ describes efficiently verifiable search: a nondeterministic
machine can guess a witness and check it in polynomial time, but an
efficient deterministic algorithm for finding a witness is not required.
Since there are
only $2^{b_f(|x|)+1}-1$ binary strings through the witness bound, every
fiber is finite. This finiteness need not give a practical enumeration
algorithm, or a constant upper bound shared by all instances.
The class $\TFNP$ consists of the total relations in $\FNP$, that is,
those with domain $\Sig^*$~\cite[Definition~2.6]{EGG25}.

\subsection{Refinements and one-query comparisons}\label{sec:reductions}

A relation can allow several answers. To solve it, an algorithm chooses
one whenever possible and reports failure otherwise. Formally, a
\emph{refinement} of $f$ is a total function
$f^*:\Sig^*\to\Sig^*\cup\{\bot\}$ with $f^*(x)\in f(x)$ when the
fiber is nonempty and $f^*(x)=\bot$ exactly when it is empty. The symbol
$\bot$ is a special no-answer output, not a binary witness. Even when
$\varepsilon$ is valid, it differs from $\bot$. We call a relation
\emph{polynomial-time solvable} if it has a polynomial-time refinement.
This requires correct behavior on empty fibers as well as finding witnesses
on nonempty ones.

We write $f\lmT g$ if a deterministic polynomial-time oracle algorithm
computes a refinement of $f$ using at most one query to any refinement of
$g$. A query names an instance of $g$. The answer can be any witness in
that fiber; the reduction must succeed for every such answer and every
choice of target refinement. If the queried fiber is empty, the oracle
returns $\bot$, which the reduction may use. This does not give the
reduction a separate polynomial-time test for emptiness. Correctness is
required on every raw binary input, including malformed operator inputs.
The time spent writing the query, reading the oracle's answer, and
writing the output counts toward the reduction's polynomial running-time
bound; the time the oracle spends computing its answer does not.

One-query reductions between $\FNP$ relations compose. Given reductions
from $f$ to $g$ and from $g$ to $h$, simulate the first and, if it asks a
query, simulate the second on that query. The second simulation either
answers locally or makes its single query to $h$; use its decoded answer
to resume the first. This makes at most one final oracle query. For every
refinement of $h$, the inner simulation supplies a legal refinement answer
for $g$, including $\bot$ on an empty fiber, so the outer simulation is
correct. Composing polynomial query bounds with the witness bounds of
$g,h\in\FNP$ bounds all answers and outputs; the two clocked simulations
therefore have one polynomial running-time bound. No-query halts remain
no-query halts. This proves transitivity.

The notation $f\emT g$ means reducibility in both directions. An
\emph{$\FNP$ degree} is a one-query equivalence class represented by a
relation in $\FNP$, written $\dg(f)$. The least degree consists of polynomial-time-solvable relations,
including the empty relation.
A relation is \emph{$\FNP$-complete} if it belongs to $\FNP$ and every
relation in $\FNP$ reduces to it under $\lmT$.

For languages, we write $L\lm H$ if there is a total polynomial-time
computable map $h$ such that $x\in L$ if and only if $h(x)\in H$.
This is called a polynomial-time many-one reduction from $L$ to $H$.
Write $L\eqm H$ when both $L\lm H$ and $H\lm L$ hold.
Thus a many-one reduction transforms a membership question for one
language into an equivalent membership question for another.
The class $\NP$ consists of languages whose members have
polynomial-length witnesses that can be verified in deterministic
polynomial time; nonmembers have no accepting witness.
A language is $\NP$-complete if it belongs to $\NP$ and every language
in $\NP$ is polynomial-time many-one reducible to it.

\subsection{Supplied lists and residual operators}\label{sec:lists-residuals}

To ask for another witness, we supply a list of witnesses already known.
The list is an ordered representation of the set to exclude. Repetitions
are rejected, so its length counts distinct witnesses.

\begin{convention}[List encoding]\label{con:lists}
Fix an injective encoding of finite lists of binary strings with
polynomial-time encoding and decoding. We use
\[
 \tup{w_1,\ldots,w_a}
 =1^a0\,1^{|w_1|}0w_1\cdots1^{|w_a|}0w_a.
\]
The decoder checks the number of
fields, called the \emph{arity}, and rejects malformed strings.
Decoding must read one complete valid code and leave no unused trailing
bits; we call this \emph{complete decoding}. For example, $\tup{x}$
followed by an extra bit is rejected, since its arity permits only the
already completed field. The empty string is not a list code.
\end{convention}

For fixed $r\ge1$, the input $\tup{x,y_1,\ldots,y_r}$ is
\emph{valid for $f$} if all $y_a$ belong to $f(x)$ and are pairwise
distinct. Let $\Val_r(f)$ denote the valid raw codes. On a valid input set
\[
 f_{[r]}(\tup{x,y_1,\ldots,y_r})=f(x)\setminus\{y_1,\ldots,y_r\};
\]
on every other raw string the fiber is empty. This is the fixed
\emph{additional-solution level}: $r$ counts the supplied witnesses, so
level one asks for a second solution. It follows the operator
of~\cite[Definition 2.9]{EGG25}, based on the $n$-ASP
formalization~\cite[Definitions 2.4--2.5]{YS02}.

The \emph{invalid-input totalization} $f_{[r]t}$ agrees with $f_{[r]}$
on valid inputs and has fiber $\{\varepsilon\}$ on invalid inputs.
A valid list is \emph{exhausted} when it contains every member of $f(x)$.
Such a list has an empty residual in $f_{[r]}$ and $f_{[r]t}$. Totalization
therefore does not necessarily produce a total search relation.
For example, if $f(x)=\{00,01\}$, supplying $11$ is invalid, supplying
$00$ is valid and leaves $01$, and supplying both $00,01$ is valid but
exhausted. Only the first case gets $\varepsilon$ from totalization.

Validity is polynomial-time decidable: decode, check arity, verify each
supplied witness, and compare the strings pairwise. To verify a residual
candidate, also check its $f$-membership and that it differs from every
supplied witness. Its length is at most $b_f(L)$, where $L$ is the complete
encoded-input length. For the invalid-input totalization, first test whether the supplied
list is valid. If it is invalid, accept the candidate exactly when it
is $\varepsilon$. If it is valid, accept the candidate exactly when
it belongs to $f(x)$ and differs from every supplied witness.
The verifiers for $f_{[r]}$ and $f_{[r]t}$ run in polynomial time, and every accepted candidate has
polynomially bounded length. Hence $\Val_r(f)\in\PP$ and
$f_{[r]},f_{[r]t}\in\FNP$. These verifiers check individual candidates;
they do not need to decide whether the supplied list exhausts $f(x)$.
On a valid exhausted list, every candidate is rejected.

\begin{definition}[Variable-list operators]\label{def:variable-operators}
On a completely decoded valid list $\tup{x,y_1,\ldots,y_r}$, define
\[
 f_{[\star]}(\tup{x,y_1,\ldots,y_r})
 =f(x)\setminus\{y_1,\ldots,y_r\}\quad(r\ge1).
\]
Define $\Res(f)$ by the same formula with $r\ge0$, so
$\Res(f)(\tup{x})=f(x)$. All other raw inputs have empty fibers.
Their invalid-input totalizations $f_{[\star]t}$ and $\Res(f)_t$ give
$\varepsilon$ exactly on invalid inputs and are equal to their counterparts on valid
ones. In particular, valid exhausted lists stay empty.
We write $f_{[0]}:=f$ only when raw level zero is explicitly needed.
\end{definition}

The star permits an arbitrary positive number of supplied witnesses.
The residual operator also permits none. Its input $\tup{x}$ then asks
the same search question as $f(x)$ but uses a different input string. On valid supplied sets, this is the residual search task called
\emph{AnotherSolution} in~\cite[Section~2.2.2, pp.~28--29]{Str10}
and \emph{AnotherSol} in~\cite[p.~5, Proposition~7]{CS17v2}.
We specify raw encodings and invalid-input behavior separately here.

\subsection{Tags, copies, and component selection}\label{sec:tags}

A \emph{tag} is a label stored as part of a binary code. It tells a decoder
how to interpret the rest of that code. We use tags in two different ways.
A \emph{witness-copy tag} distinguishes several copies of the same
underlying witness.
A \emph{component-selector tag} identifies which search relation is
being used. For example, to combine two relations $F_1$ and $F_2$,
an input $\langle 0,x\rangle$ asks for a witness of $F_1(x)$, whereas
$\langle 1,x\rangle$ asks for a witness of $F_2(x)$. An output carries
the same selector as its input, so the decoder knows which relation
its witness belongs to.
The following conventions specify the encodings for both uses.

\begin{convention}[Shared finite tag encoding]\label{con:blocks}
Choose a fixed finite ambient set $H$ of nonnegative integer labels.
If $H\ne\emptyset$, use
$\ell=\max(1,\lceil\log_2(\max H+1)\rceil)$ bits; if $H=\emptyset$,
use $\ell=1$. For $t\in H$, write $\bl t$ for its $\ell$-bit binary
code. One ambient set and one width are shared by every tag set in a
construction. Thus for any $T\subseteq H$ and $Y\subseteq\Sig^*$,
\[
 T\cdot Y=\{\bl t u:t\in T,\ u\in Y\}.
\]
Here juxtaposition concatenates the prefix and suffix. A constant with
label $t$ is the string $\bl t$ itself. For a label set $K\subseteq H$,
its constant strings are $\{\bl t:t\in K\}$.
For example, labels $2,3,4$ in the ambient set $\{0,1,2,3,4\}$ have
codes $010,011,100$. The subset need not start at zero. We never separately
renumber two tag sets from zero: their shared numerical codes prevent
collisions.

The map $(t,u)\mapsto\bl t u$ is injective because the fixed tag width lets us recover both $t$ and $u$ uniquely. We write $\uplus$ for disjoint union. Constants and witness copies use different tags, so they cannot be equal, even when $u=\varepsilon$ and the copy consists only of its tag. Tagging and untagging leave $\bot$ unchanged.
\end{convention}

\begin{convention}[Finite tagged sums]\label{con:tagged-sum}
Let $(F_i)_{i\in I}$ be a fixed finite family of search relations, indexed
by a fixed finite set $I$, with $k=|I|$. These are relations on instance/witness pairs,
not fibers of one instance. Fix an enumeration $i_0,\ldots,i_{k-1}$ of
$I$ and encode the selector $i_a$ by the $d$-bit code $J_{i_a}$ of its
position $a$, where $d=\max(1,\lceil\log_2 k\rceil)$ for $k\ge1$.
Define
\[
 \Bigl(\bigoplus_{i\in I}F_i\Bigr)(\tup{J_i,z})
   =\{\tup{J_i,v}:v\in F_i(z)\}.
\]
Both input and candidate output must decode into exactly two fields,
with an in-range fixed-width selector. Their selectors must match: an
output naming a different component is rejected. Complete decoding rejects
unused suffixes. Every other input has an empty fiber. For $I=\emptyset$
the sum is everywhere empty, without evaluating $\log_2 0$.
\end{convention}

The notation $F_0,\ldots,F_{k-1}$ uses $I=\{0,\ldots,k-1\}$;
indices starting at one, or any other finite labels, use the same definition.
For two components $F_1,F_2$, their selectors are $J_1=0,J_2=1$.
Input $\tup{0,z}$ asks for $F_1(z)$ and accepts only outputs
$\tup{0,v}$ with $v\in F_1(z)$. An answer $\tup{1,v}$ is invalid even
if its suffix happens to be a witness for both relations. The selector identifies which relation supplies the witness; the witness-copy tag introduced earlier identifies a copy of a witness within that relation. If a component witness is $\bl t u$, the sum output is $\langle J_i,\bl t u\rangle$: $J_i$ selects the relation $F_i$, while $t$ labels the copy of $u$. If sums are nested, successive selectors identify a component through several layers. We can replace this sequence by one selector for each final component, giving an equivalent single finite sum. For example, $(F_1\oplus F_2)\oplus F_3\emT F_1\oplus F_2\oplus F_3$: each direction uses one oracle query, with the selectors decoded and re-encoded.

Componentwise one-query equivalence of two fixed finite $\FNP$ families
gives equivalence of their tagged sums: decode the source selector, run
that component's reduction, and re-encode any witness with the source
selector. Return $\bot$ on invalid outer inputs. Adjoining everywhere-empty
components also leaves the degree unchanged. Map the old selectors in each
direction and, on an added empty component, return $\bot$ locally. Since
the families are fixed and finite, a single polynomial majorizes all
component query, answer, and computation bounds.

\subsection{Residual splitting and tagged copies}

The next lemma makes the raw-input and resource conventions explicit. Its
equivalence says that permitting zero supplied witnesses adds just one
component: the original first-solution task.

\begin{lemma}[Residual operators and the zero slice]\label{lem:res-split}
If $f\in\FNP$, then $f_{[\star]}$, $\Res(f)$, and their invalid-input
totalizations belong to $\FNP$. Finite tagged sums preserve $\FNP$, and
\[
 \Res(f)\emT f\oplus f_{[\star]}.
\]
\end{lemma}
\begin{proof}
Let $w$ be a residual-operator input, $L=|w|$, and $v$ a candidate witness. Decode $w$ completely. Require an instance field $x$, followed by at least one supplied witness for $f_{[\star]}$, or possibly none for $\operatorname{Res}(f)$. Check that every supplied witness belongs to $f(x)$ and that they are pairwise distinct. Reject if any check fails. Otherwise, accept exactly when $v\in f(x)$ and $v$ differs from every supplied witness.

There are at most $L$ supplied strings, each of length at most $L$, so these checks take polynomial time in $L+|v|+1$. Every accepted witness has length at most $b_f(|x|)\le b_f(L)$. For invalid-input totalization, perform the same validity check but accept only $v=\varepsilon$ when the input is invalid. This zero-length witness also satisfies the bound. Valid exhausted lists still have no witnesses. Thus the residual operators and their invalid-input totalizations belong to $\mathrm{FNP}$, without requiring a test for exhaustion.

We next verify membership in $\mathrm{FNP}$ for finite tagged sums. If the family is empty, reject every candidate. Otherwise, decode the input as $\langle J_i,z\rangle$ and the candidate as $\langle J_k,v\rangle$, requiring exactly two fields in each code and no unused bits. Check that both selectors have the required width and range and that $J_i=J_k$. Then accept exactly when $v\in F_i(z)$. The family is fixed and finite, so verification takes polynomial time. If the input length is $L$, then $|z|\le L$, and the output length is bounded by
\[
5+2d+2\max_i b_{F_i}(L).
\]
Hence finite tagged sums preserve $\mathrm{FNP}$.

To reduce $\operatorname{Res}(f)$ to $f\oplus f_{[\star]}$, use selector $0$ for $f$ and selector $1$ for $f_{[\star]}$. Validate the source input first, returning $\bot$ without a query if it is invalid. If the input is a zero-witness list $\langle x\rangle$, query the sum at $\langle0,x\rangle$. If it is a valid positive list with code $w$, query $\langle1,w\rangle$. In either case, the selected component has exactly the witnesses required by the source task. Delete the outer selector from any witness answer and return its second field. Return $\bot$ when the oracle returns $\bot$. This handles both empty original fibers and exhausted supplied lists.

For the reverse reduction, decode the outer sum input completely and check its selector. Return $\bot$ if this code is invalid. On input $\langle0,x\rangle$, query $\operatorname{Res}(f)$ at $\langle x\rangle$. On input $\langle1,w\rangle$, validate $w$ as a positive supplied list for $f$. Return $\bot$ if validation fails; otherwise, query $\operatorname{Res}(f)$ at $w$. Wrap any witness answer in the original selector, and preserve $\bot$. The zero-witness query recovers $f(x)$, while a positive query has the same remaining witnesses as $f_{[\star]}$. Every legal oracle answer therefore gives a correct source answer.

Each reduction makes at most one query. Validation and encoding take polynomial time in the complete source-input length, and every query has polynomial length. The target witness bounds also give polynomial bounds on every legal oracle answer and the resulting output. Reading, decoding, and writing these strings therefore take polynomial time. Neither reduction decides exhaustion before querying.
\end{proof}

Copying a witness changes the number of solutions without changing the
first-solution search degree. It can also guarantee a new answer from a
short list. Here the guarantee is stronger than mere solvability: every
valid short list has a remaining witness.

\begin{lemma}[Creating copies of each witness]\label{lem:inflate}
Let $g\in\FNP$ and fix $q\ge2$. Using the shared tag encoding for $H=\{0,\ldots,q-1\}$, define
\[
\widetilde{g}(x)=H\cdot g(x).
\]
Thus each witness of $g(x)$ has $q$ distinct tagged copies. Then $\widetilde{g}\in\FNP$, $\operatorname{dom}(\widetilde{g})=\operatorname{dom}(g)$, and $\widetilde{g}\emT g$. For every fixed $1\le r<q$, the additional-solution task $\widetilde{g}_{[r]}$ is polynomial-time solvable. Every valid supplied list of length $r$ leaves an unused copy of the underlying witness encoded by its first entry.
\end{lemma}
\begin{proof}
For each $t\in H$, let $c_t$ denote its fixed-width $\ell$-bit code. To verify a candidate witness of $\widetilde{g}(x)$, reject if it has fewer than $\ell$ bits. Otherwise, read its first $\ell$ bits as a tag and its remaining suffix as $u$. Accept exactly when the prefix is $c_t$ for some $t\in H$ and $V_g(x,u)$ accepts. This takes polynomial time, and every accepted witness has length at most $\ell+b_g(|x|)$. Hence $\widetilde{g}\in\mathrm{FNP}$. Its fiber is nonempty exactly when $g(x)$ is nonempty, so their domains agree.

For $\widetilde{g}\lmT g$, query $g$ at $x$ and prefix any witness answer with $c_0$. For $g\lmT\widetilde{g}$, query $\widetilde{g}$ at $x$ and delete the first $\ell$ bits from any witness answer. Both reductions return $\bot$ when the oracle returns $\bot$. Every legal answer is handled correctly, and the fixed tag length and witness bounds give polynomial running time.

Now fix $1\le r<q$. Decode the source input completely, require exactly $r$ supplied witnesses, and check their membership in $\widetilde{g}(x)$ and pairwise distinctness. Return $\bot$ if validation fails. Otherwise, let $u$ be the suffix of the first supplied witness. Since $u\in g(x)$, all $q$ strings
\[
\{c_tu:t\in H\}
\]
are distinct valid witnesses of $\widetilde{g}(x)$. The supplied list contains only $r<q$ witnesses, so at least one of these copies is absent. Compare the copies with the supplied list and return an absent one. This also works when $u=\varepsilon$: the copies are then the distinct tag strings themselves.

Validation and the fixed number of comparisons take polynomial time. Thus the algorithm solves $\widetilde{g}_{[r]}$ on every raw input, and every valid input has an additional witness.
\end{proof}

\subsection{A usable set of positive instances}\label{sec:padding}
In the main construction, we start with a search relation $A$ and modify its witnesses so that finding additional solutions has a prescribed difficulty. We must keep its domain unchanged: an instance has a witness in the new relation exactly when it has one in $A$. To encode a target search task, we therefore need a way to assign each target instance to a distinct instance in $\operatorname{dom}(A)$. This assignment must be computable in polynomial time, and we must be able to recognize the assigned instances and recover the target instances from them. The following padding condition provides these maps.
To place target tasks inside an existing domain, we need positive instances
we can construct and recognize. A language $X$ is \emph{$1$-paddable via
$p$} if $p:\Sig^*\to\Sig^*$ is total and injective, $p$ and its inverse
are polynomial-time computable, and $\ran(p)\subseteq X$
\cite[Definition 3.1]{EGG25}. The inverse returns $\bot$ off the image.
Consequently $\ran(p)$ is polynomial-time decidable by running the inverse.
Both maps have polynomial output bounds, since their algorithms must write
their outputs. This assumption assigns every binary string $z$ a distinct element $p(z)\in X$, with both the assignment and recovery of $z$ computable in polynomial time. The assumption lets us recognize $\operatorname{ran}(p)$ in polynomial time; it does not require us to decide membership in all of $X$. Section~\ref{sec:gamma} uses disjoint subsets of $\operatorname{ran}(p)$ as slices; they do not cover that image, since $p(\varepsilon)$ is generic.

\section{Prescribing finitely many additional-solution levels}\label{sec:gamma}
We ask whether the difficulty of finding another solution can be chosen separately for different numbers of supplied solutions. For example, can we prescribe one search degree for finding a second solution and another for finding a third, while preserving the domain and the first-solution degree?
We now prove this for any fixed finite number of levels. Given a base relation $A$ and target relations $B_1,\ldots,B_n$, we construct one relation with the same domain and first-solution degree as $A$, whose level-$j$ additional-solution task is equivalent to $B_j$ for each $1\le j\le n$. The targets can be chosen independently and may be empty. The stated predecessor theorem assumes a nonempty target~\cite[Theorem~4.6]{EGG25}; allowing empty targets extends that hypothesis of the construction. This is not a claim that the least search degree was absent from earlier work. An empty target makes the corresponding level polynomial-time solvable, including correct reporting when no additional solution exists.
\begin{theorem}[Simultaneous spectrum realization]\label{thm:gamma}
Fix $n\ge1$. Let $A,B_1,\ldots,B_n\in\FNP$, and suppose that $\dm(A)$
is $1$-paddable via $p$. There is a relation
$C\in\FNP$ such that
\[
 \dm(C)=\dm(A),\qquad C\emT A,\qquad
 C_{[j]}\emT B_j\quad(1\le j\le n).
\]
\end{theorem}

The equality of domains concerns raw instances of $C$ and $A$.
The additional-solution tasks instead take encoded lists. Each equivalence
uses at most one query and must decode every legal oracle answer, including
$\bot$.

\subsection{Why copies and separate instance regions are needed}

The relations $B_1,\ldots,B_n$ are the target search problems given in the theorem. We want finding another witness of $C$ after one supplied witness to be equivalent to $B_1$, after two supplied witnesses to be equivalent to $B_2$, and so on. Thus $B_r$ is the chosen target for the task with exactly $r$ supplied witnesses.
Fix $i\in\{1,\ldots,n\}$. To encode $B_i$, we assign each target instance $z$ an instance $x$ of $C$. Suppose first that $C(x)$ contains $i$ explicit constants and one tagged copy of each witness in $B_i(z)$, with no other witnesses. Supplying all $i$ constants leaves exactly the target copies. Finding another witness then amounts to finding a witness in $B_i(z)$.
The problem appears when more than $i$ witnesses are supplied. For example, suppose $i<n$ and the list contains all $i$ constants together with the copy of one target witness $u$. This is a list of length $i+1$. An additional witness must encode a different witness of $B_i(z)$, so the remaining task is finding a second solution of $B_i$. But the theorem requires the task with $i+1$ supplied witnesses to have the degree of the independently chosen relation $B_{i+1}$. There is no reason for these two degrees to agree.
Following the witness-copy argument of Lemma~\ref{lem:inflate}, we therefore give each witness in $B_i(z)$ exactly $n+1$ tagged copies. A supplied list of length at most $n$ cannot contain all copies of any target witness. If it contains a copy of $u$, we can return an absent copy of that same $u$, without finding a different target witness. If the list contains only constants but omits one, we return the omitted constant. The only remaining case is the list containing all $i$ constants and no copies: it has length $i$ and leaves exactly the copies of $B_i(z)$. Thus this instance encodes the desired target search at list length $i$, while every valid list of another length between $1$ and $n$ can be extended in polynomial time.

We now specify which input strings $x$ of the constructed relation $C$ will be used for each target problem. We choose pairwise disjoint sets $S_0,\ldots,S_n\subseteq\dm(A)$ and call each set a \emph{slice}. For $1\le i\le n$, inputs $x\in S_i$ will encode instances of $B_i$. The set $S_0$ will be used to preserve the first-solution degree of $A$. Thus a slice is a set of input strings $x$, whereas $C(x)$ is the set of witnesses for one such input. The padding map $p$ (from the $1$-paddability assumption on $\dm(A)$) lets us construct and recognize these slices in polynomial time. Input strings outside their union will be called \emph{generic}.

\begin{lemma}[Slices]\label{lem:slices}
Fix $n\ge1$ and let $c_0,\ldots,c_n$ be the distinct binary codes of
$0,\ldots,n$ of length $\ell_0=\lceil\log_2(n+1)\rceil$.
Put $p_i(z)=p(c_i z)$ and
\[
 G=\Sig^*\setminus\bigcup_{i=0}^n\ran(p_i).
\]
The images of the $p_i$ are pairwise disjoint polynomial-time decidable
subsets of $\dm(A)$ when $p$ pads $\dm(A)$. The maps $p_i$ and their
inverses are polynomial-time computable. Also
$p(\varepsilon)\in G$ and $\Sig^*\setminus\ran(p)\subseteq G$.
\end{lemma}
\begin{proof}
Given a string $x$, compute $p^{-1}(x)$. If the result is $\bot$ or has fewer than $\ell_0$ bits, then $x$ belongs to none of the images $\operatorname{ran}(p_i)$. Otherwise,
\[
x\in\operatorname{ran}(p_i)
\quad\Longleftrightarrow\quad
p^{-1}(x)=c_i z
\]
for some suffix $z$. Indeed, this equality is equivalent to $x=p(c_i z)=p_i(z)$. Deleting the prefix $c_i$ therefore recovers $p_i^{-1}(x)=z$. Computing $p^{-1}$ and checking the prefix take polynomial time, so each image is polynomial-time decidable and each inverse is polynomial-time computable. Each $p_i$ is also polynomial-time computable by its definition.

The images are pairwise disjoint because $p$ is injective and the codes $c_i$ are distinct and have equal length. Every image lies in $\operatorname{ran}(p)\subseteq\operatorname{dom}(A)$. Finally, $p^{-1}(p(\varepsilon))=\varepsilon$ has no $\ell_0$-bit prefix, so $p(\varepsilon)\in G$. Strings outside $\operatorname{ran}(p)$ also belong to no image, and hence lie in $G$.
\end{proof}

We use these slice maps for $\gamma$ and its arbitrary-list analysis. The later counter-copy construction will specify its own slices.

\subsection{Finite construction and proof}\label{sec:finite-proof}

Fix $n,A,B_1,\ldots,B_n$ and their common padding $p$ as in the theorem.
A branch slice $\ran(p_i)$, $i\ge1$, will carry the target $B_i$ at
level $i$. Replacing its original $A$ witnesses could lose information
needed to solve $A$. The \emph{repair slice} $\ran(p_0)$ restores that
access by storing copies of $A(p(z))$ for every $z$. Thus every original
instance inside $\ran(p)$ can be moved to a place retaining its original
witnesses. Instances outside the slices, in $G$, remain generic.

Use one ambient tag set $H=\{0,\ldots,2n\}$ in Convention~\ref{con:blocks}. For each $0\le i\le n$, define the set $K_i$ of constant strings and the set $T_i$ of copy labels by
\[
K_i=\{\bl t:0\le t<i\},\qquad T_i=\{i,\ldots,i+n\}.
\]
In particular $K_0=\emptyset$. The labels of the constants $K_i$ are disjoint
from $T_i$, and all codes use the same width
$\ell=\lceil\log_2(2n+1)\rceil$. Each $T_i$ has $n+1$ labels.
Set $B_0(z)=A(p(z))$. This auxiliary relation is total because padding
lands in $\dm(A)$; it need not have a polynomial-time refinement.

\begin{definition}[Spectrum construction]\label{def:gamma}
Let $\gamma=\gamma_n(A,B_1,\ldots,B_n)$ be given by
\[
 \gamma(x)=
 \begin{cases}
 T_0\cdot A(x),&x\in G,\\[2pt]
 K_i\uplus\bigl(T_i\cdot B_i(p_i^{-1}(x))\bigr),
       &x\in\ran(p_i),\quad 0\le i\le n.
 \end{cases}
\]
The padding $p$ is fixed but suppressed in the notation. Slice zero (for $p_0$) is the
repair slice; the other slices are branches. The disjoint union separates
constant strings from tagged witness copies, including copies of
$\varepsilon$.
\end{definition}

\begin{lemma}[Fibers and admissibility]\label{lem:fiber}
The relation $\gamma$ belongs to $\FNP$. Its fibers have the forms
\[
 \begin{array}{c|c}
 \text{instance} & \text{fiber}\\ \hline
 x\in G & T_0\cdot A(x)\\
 x=p_0(z) & T_0\cdot A(p(z))\\
 x=p_i(z),\ 1\le i\le n & K_i\uplus(T_i\cdot B_i(z)).
 \end{array}
\]
In the last row, $|\gamma(x)|=i+(n+1)|B_i(z)|$.
The repair and branch fibers are nonempty, whether or not $B_i(z)$ is empty.
\end{lemma}
\begin{proof}
Given an instance $x$ and a candidate witness $v$, first determine whether $x$ is generic, on the repair slice, or on branch $i$, using Lemma~\ref{lem:slices}. On branch $i\ge1$, compute $z=p_i^{-1}(x)$ and accept if $v$ is one of the constants in $K_i$. Otherwise, require at least $\ell$ bits, read the first $\ell$ bits as a tag, and let $u$ be the remaining suffix. Accept exactly when the tag belongs to $T_i$ and $u\in B_i(z)$.

On generic and repair instances, require a tag in $T_0$ and verify the suffix in $A(x)$ or $A(p(p_0^{-1}(x)))$, respectively. Reject candidates with missing or inadmissible tags. The slice tests and instance transformations are polynomial-time computable. Since only finitely many fixed graph verifiers are used, the resulting verifier runs in polynomial time in $|x|+|v|+1$.

Choose a monotone polynomial $\rho$ bounding the lengths of $x$, every inverse image $p_i^{-1}(x)$, and the repair instance $p(p_0^{-1}(x))$ whenever these are used. Every copied witness then has length at most
\[
\ell+\max\left\{
b_A(\rho(|x|)),
\max_{1\le i\le n}b_{B_i}(\rho(|x|))
\right\}.
\]
Constants have length $\ell$, so the same bound covers them. Together with the graph verifier, this proves $\gamma\in\mathrm{FNP}$.

The displayed fibers follow directly from the definition, with $B_0(z)=A(p(z))$ substituted on the repair slice. On branch $i$, each underlying witness has exactly $n+1$ distinct tagged copies. Different underlying witnesses give different strings, and the copy tags are disjoint from the constant labels, even for an empty suffix. Hence
\[
|\gamma(p_i(z))|=i+(n+1)|B_i(z)|.
\]
Every repair fiber is nonempty because $p(z)$ belongs to the domain of $A$. Every branch fiber contains its $i\ge1$ constants, including when $B_i(z)$ is empty.
\end{proof}

Given a valid supplied list, we first try to produce another witness in either of two ways: return a constant missing from the list, or take a supplied copy $\bl t u$ and return another copy $\bl s u$ of the same witness $u$ that is missing from the list. We call these the two \emph{elementary rules}. In the second rule, the tag $s$ may already appear on a copy of a different witness; only the complete string $\bl s u$ must be absent. Neither rule applies exactly when the list contains every constant and, for each underlying witness represented in the list, all its tagged copies. We call such a supplied list a \emph{full grid}. The next lemma states this characterization formally.

\Needspace{10\baselineskip}
\begin{lemma}[Full grids]\label{lem:grid}
Use one shared finite tag encoding. Let $K$ be a finite set of constant strings whose labels are disjoint from a nonempty finite set $T$ of copy labels. Let $Y\subseteq\Sig^*$ and put $F=K\uplus(T\cdot Y)$, the set of available witnesses. Let $S\subseteq F$ be a finite set of supplied witnesses, and define
\[
U=\{u:\bl t u\in S\text{ for some }t\in T\}.
\]
Thus $U$ contains the underlying witnesses whose copies appear in $S$. Consider the two elementary ways to find a witness in $F\setminus S$: return a constant in $K\setminus S$, or return an absent copy $\bl t u$ with $t\in T$ and $u\in U$. Neither way produces a witness exactly when
\[
S=K\uplus(T\cdot U),
\]
that is, when $S$ contains every constant and every copy of each witness in $U$. Equivalently, $K\subseteq S$ and $|S|=|K|+|T||U|$. In that case,
\[
F\setminus S=T\cdot(Y\setminus U).
\]
Thus any remaining witness must be a copy of an underlying witness not yet represented in $S$.
\end{lemma}
\begin{proof}
The first rule finds another witness exactly when a constant is missing from $S$. It therefore fails exactly when $K\subseteq S$. The second rule fails exactly when every tagged copy of every underlying witness in $U$ is already in $S$, that is, when $T\cdot U\subseteq S$. By the definition of $U$, every supplied witness that is not a constant belongs to $T\cdot U$. Thus
\[
S\subseteq K\uplus(T\cdot U).
\]
Both rules fail exactly when the reverse inclusion also holds, giving
\[
S=K\uplus(T\cdot U).
\]
For the cardinality formulation, assume $K\subseteq S$. Then $S\setminus K\subseteq T\cdot U$, and
\[
|S\setminus K|=|S|-|K|,
\qquad
|T\cdot U|=|T||U|.
\]
The second equality follows from injectivity of the tag encoding. These finite sets are equal exactly when their cardinalities agree, so the full-grid condition is equivalent to $K\subseteq S$ and $|S|=|K|+|T||U|$.

Finally, removing a full grid removes every constant and every copy of each witness in $U$. The remaining witnesses are therefore exactly
\[
F\setminus S=T\cdot(Y\setminus U).
\]
If $U=Y$, this set is empty: failure of both rules does not itself guarantee that another witness exists.
\end{proof}

On branch slice $i$, a full grid has size $i+(n+1)|U|$.
On a generic or repair fiber it has size $(n+1)|U|$.
Now consider a valid supplied list of length $r$, where $1\le r\le n$, and let $S$ be its set of witnesses. If both elementary rules fail, Lemma~\ref{lem:grid} says that $S$ contains every constant and all $n+1$ copies of each underlying witness in $U$. Since $|S|=r\le n$, the list cannot contain a complete group of $n+1$ copies. Therefore $U=\emptyset$, so $S$ consists only of constants.
Generic and repair instances have no constants. On either type of instance, such a full grid would therefore be empty, contradicting $r\ge1$. On branch slice $i$, there are exactly $i$ constants, so $S=K_i$ and $r=i$. Thus both rules can fail only when the instance lies in $\ran(p_r)$ and the supplied list contains exactly $K_r$. In that case the remaining witnesses are $T_r\cdot B_r(z)$, where $z=p_r^{-1}(x)$. This is precisely the case in which finding another witness requires the target search $B_r$; every other valid list of length $r$ can be extended by an elementary rule.

\begin{proof}[Proof of Theorem~\ref{thm:gamma}]
Take $C=\gamma$. Lemma~\ref{lem:fiber} gives $\gamma\in\mathrm{FNP}$. On a generic instance $x\in G$, the fiber is nonempty exactly when $A(x)$ is nonempty. Every repair or branch instance belongs to $\operatorname{ran}(p)\subseteq\operatorname{dom}(A)$ and has a nonempty $\gamma$ fiber. Hence
\[
\operatorname{dom}(\gamma)=\operatorname{dom}(A).
\]

For $\gamma\lmT A$, first classify the instance. On a branch, return the constant with label zero without querying. On a generic instance, query $A$ at $x$. On a repair instance, query $A$ at $p(p_0^{-1}(x))$. In either query case, prefix a witness answer with tag zero and preserve $\bot$. The fiber formulas show that every returned witness belongs to $\gamma(x)$.

For the reverse reduction, we must avoid branch queries that could return constants instead of original $A$ witnesses. Define
\begin{equation}\label{eq:repair-map}
h(x)=
\begin{cases}
x,&x\notin\operatorname{ran}(p),\\
p_0(p^{-1}(x)),&x\in\operatorname{ran}(p).
\end{cases}
\end{equation}
If $x$ is outside the padding image, it is generic. Otherwise, $h(x)$ is a repair instance whose underlying $A$-instance is $p(p^{-1}(x))=x$. Thus, in both cases,
\[
\gamma(h(x))=T_0\cdot A(x).
\]
Query $\gamma$ at $h(x)$ and delete the tag from any witness answer. Return $\bot$ on an oracle answer $\bot$. Every witness answer yields a member of $A(x)$, and the queried fiber is empty exactly when $A(x)$ is empty. This proves $\gamma\emT A$.

Fix $1\le j\le n$. Decode the source input completely, require exactly $j$ supplied witnesses, and check their membership in $\gamma(x)$ and pairwise distinctness. Return $\bot$ if validation fails. Otherwise, let $S$ be the supplied set. Try the two elementary rules: return a missing constant, or an absent copy of an underlying witness already represented in $S$. Either successful rule gives a witness in $\gamma(x)\setminus S$ without a query. If both fail, $S$ is a full grid. It cannot contain a copied witness, since a full grid would then contain all $n+1>j$ copies of that witness. Thus $S$ consists only of constants. Generic and repair instances have no constants, while branch $i$ has exactly $i$ constants. Consequently, $x=p_j(z)$ and $S=K_j$, leaving
\[
\gamma(x)\setminus S=T_j\cdot B_j(z).
\]
Query $B_j$ at $z$ and prefix any witness answer with tag $j$, which belongs to $T_j$. Preserve $\bot$. If $B_j(z)$ is empty, the supplied constants exhaust the source fiber, so $\bot$ is also the correct source answer.

Given any input $z$, query $\gamma_{[j]}$ on the encoded instance $p_j(z)$ followed by all $j$ constants in $K_j$. These constants are valid and distinct, and the remaining witnesses are exactly $T_j\cdot B_j(z)$. Delete the tag from any witness answer to obtain a member of $B_j(z)$, and preserve $\bot$. This proves the reverse reduction.

All reductions make at most one query. Slice classification, validation, padding, and the fixed constant and copy tests take polynomial time in the complete source-input length. Query lengths are polynomially bounded, and composing these bounds with the target witness bounds also bounds every legal oracle answer and decoded output. Thus answer processing takes polynomial time. Every legal oracle answer is handled correctly, including $\bot$, without a separate test for exhaustion.
\end{proof}

Each underlying witness has exactly $n+1$ tagged copies. A supplied list of at most $n$ witnesses therefore leaves an unused copy of every underlying witness represented in the list. Longer lists can contain complete groups of copies. If a list contains every constant and all copies of every underlying witness represented in it, neither elementary rule gives another solution. The full-grid lemma then shows that the remaining solutions are exactly the copies of underlying witnesses not represented in the list. Section~\ref{sec:variable} uses this characterization to analyze arbitrary-length supplied lists.

\subsection{Prescribing finite tuples of search degrees}\label{sec:cor}

The fixed-level theorem gives two direct degree consequences. Without a
specified domain, every finite tuple of $\FNP$ degrees can occur. On a
specified $1$-paddable domain, the first degree must already be attainable
by some relation on that domain; the additional degrees can then be chosen
freely.

\begin{corollary}[Prescribing search degrees]\label{cor:C3}
Fix $n\ge1$. Given any $\FNP$ degrees $d_0,d_1,\ldots,d_n$, there is a relation $C\in\FNP$ such that
\[
\dg C=d_0,\qquad \dg C_{[j]}=d_j\quad(1\le j\le n).
\]
Moreover, let $X\in\NP$ be a fixed $1$-paddable language. If some relation $A\in\FNP$ has domain $X$ and degree $d_0$, then $C$ can be chosen with domain $X$ for any prescribed $\FNP$ degrees $d_1,\ldots,d_n$.
\end{corollary}
\begin{proof}
Choose relations $A,B_1,\ldots,B_n\in\mathrm{FNP}$ representing the degrees $d_0,d_1,\ldots,d_n$. Following~\cite[Proposition~3.6]{EGG25}, define
\[
\widehat{A}(0x)=A(x),\qquad
\widehat{A}(1x)=\{\varepsilon\},\qquad
\widehat{A}(\varepsilon)=\{\varepsilon\}.
\]
A verifier uses the verifier of $A$ after an initial zero and accepts only the witness $\varepsilon$ on every other input. The witness bound of $A$ also bounds the copied witnesses, so $\widehat{A}\in\mathrm{FNP}$.

To reduce $A$ to $\widehat{A}$, query at $0x$ and return the oracle answer unchanged. Conversely, on input $0x$, query $A$ at $x$; on every other input, return $\varepsilon$ without querying. These reductions preserve $\bot$ on empty fibers, proving $\widehat{A}\emT A$.

The domain of $\widehat{A}$ is $1$-paddable via $z\mapsto1z$: every such input has witness $\varepsilon$, and the inverse deletes the initial one, returning $\bot$ on strings outside the image. Thus $\widehat{A}$ represents $d_0$ and satisfies the padding assumption. Theorem~\ref{thm:gamma}, applied to $\widehat{A},B_1,\ldots,B_n$, realizes the prescribed tuple.

For the fixed-domain claim, no normalization is needed. Apply the theorem directly to the given relation $A$ with domain $X$ and to representatives of the prescribed additional-solution degrees.
\end{proof}

The fixed-domain claim requires the first-solution degree to be realizable on $X$. For example, if a relation is polynomial-time solvable, its domain belongs to $\PP$: run its solver and accept exactly when the output is not $\bot$. Consequently, if $X\notin\PP$, no relation with domain $X$ can have the least search degree. Once a first-solution degree is realizable on a $1$-paddable domain, all finitely many additional-solution degrees can be chosen freely.

\section{Finding another solution after an arbitrary list}\label{sec:variable}

\begin{lemma}[Fixed levels reduce to arbitrary lists]\label{lem:fixed-star}
For every $f\in\FNP$ and every fixed $j\ge1$,
\[
 f_{[j]}\lmT f_{[\star]}.
\]
Consequently, for fixed $n\ge1$,
$\bigoplus_{j=1}^n f_{[j]}\lmT f_{[\star]}$.
\end{lemma}
\begin{proof}
Decode the source input completely and require exactly $j$ supplied
witnesses, all in $f(x)$ and pairwise distinct. Return $\bot$ locally
if any check fails. On a valid input, query $f_{[\star]}$ with that
same list and return its answer unchanged. The list has positive arity,
and the two residual fibers are equal, including when exhausted.
Validation, query writing, and every legal answer are polynomially bounded
in the complete input length. For the finite sum, first decode the
selector and then perform the selected level's reduction, wrapping a
witness answer in its original selector and preserving $\bot$. A fixed
finite family has one polynomial resource bound.
\end{proof}

Thus any prescribed arbitrary-list degree must dominate every prescribed
fixed-level degree. The compatibility question is whether this necessary
condition suffices while the domain and first-solution degree are retained.

\subsection{Arbitrary-list search for the finite construction}

At a fixed level the number of supplied witnesses is a constant of the
problem. Here it is part of the input. We therefore need one algorithm
whose running time is polynomial in the entire encoded list, for every
possible arity. A collection of separately polynomial fixed-level algorithms would not suffice: their constants could grow too quickly. For each fixed $r$, a running time such as $2^rL^2$ is polynomial in the encoded-input length $L$, because $2^r$ is then a constant. When $r$ varies with the input, this bound need not be polynomial in $L$; the fixed-level degrees do not guarantee one algorithm with a polynomial bound valid for all $r$.

Here $A$ is the original search relation whose domain and first-solution degree the construction $\gamma$ preserves, and $B_1,\ldots,B_n$ are the chosen target relations. Generic and repair instances of $\gamma$ have only tagged copies of witnesses from $A$, with no constants. Branch $j$ has $j$ constants together with tagged copies of witnesses from $B_j$.
Recall that $f_{[\star]}$ requires at least one supplied witness, whereas $\Res(f)$ also permits none. When both elementary rules fail, the supplied list contains every constant and all copies of each underlying witness represented in it. On a generic or repair instance, a nonempty list must therefore represent at least one witness from $A$. After removing the copy tags and repetitions, we obtain a positive supplied list for $A_{[\star]}$.
On branch $j$, however, the supplied list may contain only the $j$ constants. Removing them leaves no known witness of $B_j(z)$, so the remaining task is to find a first witness of $B_j(z)$. Other full grids may also supply underlying $B_j$ witnesses. The operator $\Res(B_j)$ covers both cases, allowing either zero or more supplied witnesses.

The next theorem characterizes finding another witness of $\gamma$ after an arbitrary positive-length supplied list. Its degree is determined by additional-solution search for the base relation $A$ and residual search for the target relations $B_j$. The tagged sum represents a choice among these tasks: an input specifies one component, and an answer solves that component. A reduction to this sum therefore needs at most one oracle query.
\begin{theorem}[Arbitrary-list search for the finite construction]\label{thm:gamma-star}
Fix $n\ge1$ and $A,B_1,\ldots,B_n\in\FNP$. Suppose that $\dm(A)$ is $1$-paddable via $p$, and let $\gamma$ be the relation constructed from these data in Definition~\ref{def:gamma}. Then
\[
\gamma_{[\star]}\emT A_{[\star]}\oplus\bigoplus_{j=1}^n\Res(B_j).
\]
In each direction, a single reduction works for every supplied-list length, uses at most one oracle query, and runs in time bounded by one polynomial in the complete source-input length. The number $n$, the relations, and the padding map are fixed. The target relations $B_j$ may be empty or partial.
\end{theorem}
\begin{proof}
Put $q=n+1$. To reduce $\gamma_{[\star]}$ to the residual sum, decode the source input completely, require a nonempty supplied list, and check that its entries are distinct witnesses of $\gamma(x)$. Return $\bot$ without querying if validation fails. Otherwise, let $S$ be the supplied set and classify $x$ as generic, repair, or branch.

Try the two elementary rules: return a missing branch constant, or an absent copy of an underlying witness already represented in $S$. If either succeeds, return that witness without querying. Otherwise, let $U$ be the set of distinct suffixes of the supplied copies, ignoring constants. Lemma~\ref{lem:grid} gives, on branch $j$ with $x=p_j(z)$,
\[
S=K_j\uplus(T_j\cdot U),\qquad
\gamma(x)\setminus S=T_j\cdot(B_j(z)\setminus U).
\]
On a generic or repair instance, put $a=x$ or $a=p(p_0^{-1}(x))$, respectively. Then
\[
S=T_0\cdot U,\qquad
\gamma(x)\setminus S=T_0\cdot(A(a)\setminus U).
\]
In this case $U$ is nonempty because $S$ is nonempty and there are no constants.

List the elements of $U$ in lexicographic order as $u_1,\ldots,u_m$. On branch $j$, query the tagged sum in component $\operatorname{Res}(B_j)$ at $\langle z,u_1,\ldots,u_m\rangle$; here $m$ may be zero. On a generic or repair instance, query component $A_{[\star]}$ at $\langle a,u_1,\ldots,u_m\rangle$; here $m\ge1$. These inner lists are valid because their entries come from valid supplied copies.

Remove the component selector fields from any witness answer to obtain $v\notin U$. Prefix $v$ with tag $j$ on a branch or tag zero on a generic or repair instance. The displayed residual formulas show that this produces an unsupplied witness of $\gamma(x)$. They also show that the queried component is empty exactly when the source residual is empty, so an oracle answer $\bot$ is correctly returned as $\bot$.

For the reverse reduction, decode the source input and its component selector fields completely. Validate the inner list for the selected operator: require at least one supplied witness for $A_{[\star]}$, allow none for $\operatorname{Res}(B_j)$, and check all supplied memberships and pairwise distinctness. Return $\bot$ without querying if any check fails. A zero-witness list for $\operatorname{Res}(B_j)$ is valid even when $B_j(z)$ is empty.

For a valid $A_{[\star]}$ input $\langle a,u_1,\ldots,u_m\rangle$, query $\gamma_{[\star]}$ at $h(a)$ from~\eqref{eq:repair-map}, supplying all $q$ copies of every $u_i$ with tags in $T_0$. Since $m\ge1$, the expanded list is nonempty. The identity $\gamma(h(a))=T_0\cdot A(a)$ gives the residual
\[
T_0\cdot(A(a)\setminus\{u_1,\ldots,u_m\}).
\]
For a valid $\operatorname{Res}(B_j)$ input $\langle z,u_1,\ldots,u_m\rangle$, query at $p_j(z)$, supplying all constants in $K_j$ and all $q$ copies of every $u_i$ with tags in $T_j$. The list has $j+qm\ge1$ entries and leaves exactly
\[
T_j\cdot(B_j(z)\setminus\{u_1,\ldots,u_m\}).
\]
Both expanded lists are valid: their entries are witnesses, and the tag encoding makes them distinct.

Delete the copy tag from any oracle witness and wrap its suffix in the original component selector fields. This gives a witness of the source sum. Preserve $\bot$, since the source and queried residuals are empty in exactly the same cases. This includes empty suffixes, exhausted lists, and zero-witness branch lists.

For a uniform polynomial bound, let $L$ be the complete source-input length. Validation examines at most $L$ supplied strings. Forward extraction and sorting process at most $L$ suffixes; reverse expansion writes at most $n+qL$ witnesses. Since $n$ and $q$ are fixed, these operations and the padding transformations take polynomial time in $L$. Query lengths are polynomially bounded, and the fixed target witness bounds also bound every legal oracle answer and decoded output polynomially. Thus each direction uses at most one query and one polynomial running-time bound valid for every supplied-list length. Neither direction requires a separate test for exhaustion.
\end{proof}

Every legal oracle answer has polynomially bounded length, so reading and decoding it take polynomial time. Both reductions handle invalid inputs and exhausted lists. A branch list containing only constants leaves no supplied target witness, which is why $\Res(B_j)$ must allow an empty list.
Thus the theorem characterizes arbitrary-list search for this finite construction in terms of $A_{[\star]}$ and $\Res(B_1),\ldots,\Res(B_n)$. It does not prescribe its degree independently of these component tasks.

\subsection{Prescribing a compatible arbitrary-list degree}\label{sec:delta}

The finite construction uses a constant number of copies and retains the
component residual tasks in Theorem~\ref{thm:gamma-star}. We now give a
different construction. Each underlying witness has exponentially many
copies, identified by polynomial-length binary counters. A short list
cannot exhaust these copies. A list that does exhaust them is long enough
to pay for checking every possible underlying witness. Exponential
replication and this use of accumulated output length already occur in
the enumeration construction $\mathrm{Pad}_t$ of~\cite[Section~4, p.~11,
Lemma~17]{CS17v2}. Here they allow a prescribed arbitrary-list degree to
coexist with finitely many prescribed fixed-level degrees on a fixed domain.

\begin{theorem}[Compatible fixed-level and arbitrary-list realization]
\label{thm:delta}
Fix $n\ge1$ and let $A,B_1,\ldots,B_n,B\in\FNP$. Suppose that
$\dm(A)$ is $1$-paddable and $B_j\lmT B$ for $1\le j\le n$.
There is a relation $C\in\FNP$ such that
\[
 \dm(C)=\dm(A),\qquad C\emT A,\qquad
 C_{[j]}\emT B_j\ (1\le j\le n),\qquad C_{[\star]}\emT B.
\]
No assumption about $\PP$ versus $\NP$ is required.
\end{theorem}

Fix a padding map $p$. The slice maps in this subsection are new and
belong only to the new construction: use distinct fixed-width codes
$c_0,\ldots,c_{n+1}$ and set $p_i(z)=p(c_i z)$. Write
$G_\delta=\Sig^*\setminus\bigcup_{i=0}^{n+1}\ran(p_i)$.
The proof of Lemma~\ref{lem:slices} gives recognizable disjoint slices
with polynomial-time inverses. They are subsets of $\ran(p)$ and do not
cover it: $p(\varepsilon)$ is generic. Slice zero is repair; set
$D_i=B_i$ for $1\le i\le n$ and $D_{n+1}=B$.

Choose a monotone nonnegative integer-valued polynomial $\rho$ bounding
$|x|$, all $|p_i^{-1}(x)|$ when defined, and
$|p(p_0^{-1}(x))|$ on repair instances. Such a polynomial exists because
these maps are polynomial-time computable and must write their outputs.
Using monotone witness bounds for the fixed relations, choose a monotone
nonnegative integer-valued polynomial $b_*$ majorizing
\[
 b_A(\rho(m))+\sum_{i=1}^{n+1}b_{D_i}(\rho(m)).
\]
It therefore bounds the length of every underlying witness used at any
new instance $x$ of length $m$: witnesses of $A(x)$, of
$A(p(p_0^{-1}(x)))$, and of $D_i(p_i^{-1}(x))$.
The composition with $\rho$ is essential; the target witness bound at
the original input $z$ alone need not be a bound in $|x|$.

Use one tag width $\ell=\kappa=\max(1,\lceil\log_2(n+2)\rceil)$
for labels $0,\ldots,n+1$. Put
\[
 s(m)=b_*(m)+\kappa,\qquad Q_x=2^{s(|x|)},\qquad
 K_i=\{\bl t:0\le t<i\}.
\]
For $Y\subseteq\Sig^*$ define
\[
 \operatorname{Copy}_x(i,Y)
   =\{\bl i\,c\,u:c\in\Sig^{s(|x|)},\ u\in Y\}.
\]
The counter $c$ has the width determined by the actual instance $x$.
In particular,
\begin{equation}\label{eq:delta-count}
 Q_x\ge n+2,\qquad 2^{b_*(|x|)+1}-1<Q_x.
\end{equation}

\begin{definition}[Counter-copy construction]\label{def:delta}
Define $\delta=\delta_n(A,B_1,\ldots,B_n;B)$ by
\[
 \delta(x)=
 \begin{cases}
 \operatorname{Copy}_x(0,A(x)),&x\in G_\delta,\\[2pt]
 \operatorname{Copy}_x(0,A(p(z))),&x=p_0(z),\\[2pt]
 K_i\uplus\operatorname{Copy}_x(i,D_i(z)),
       &x=p_i(z),\quad 1\le i\le n+1.
 \end{cases}
\]
This definition leaves $\gamma$ unchanged. One copy tag and one
variable-width counter encode each copy; no exponential family is listed
as part of the construction.
\end{definition}

\begin{proof}[Proof of Theorem~\ref{thm:delta}]
Take $C=\delta$. To verify its graph, classify $x$ by the recognizable
slices. On branch $i$, accept the constants $K_i$. Otherwise, require at
least $\ell+s(|x|)$ bits, check the copy tag (zero on generic and repair
instances, $i$ on branch $i$), parse the next $s(|x|)$ bits as the counter,
and verify the remaining suffix in the appropriate underlying relation.
The finitely many verifiers and polynomial-time instance transformations
give one polynomial verifier bound in $|x|+|v|+1$. All witnesses have length
at most $\ell+s(|x|)+b_*(|x|)$. Constants have length $\ell$, whereas
copies have length at least $\ell+s(|x|)>\ell$ and use a tag distinct from
the branch's constant labels. Thus the union is disjoint even for suffix
$\varepsilon$, and $\delta\in\FNP$.

On a generic input, $\delta(x)$ is nonempty exactly when $A(x)$ is.
Repair fibers are nonempty since $p(z)\in\dm(A)$; branch fibers contain
their constants. All slices lie in $\dm(A)$, proving domain equality.
For $\delta\lmT A$, return a constant on a branch, query $A(x)$ on a
generic input, and query $A(p(p_0^{-1}(x)))$ on repair. Encode any witness
answer using tag zero and the all-zero counter of width $s(|x|)$, preserving
$\bot$. For $A\lmT\delta$, use the repair map
\[
 h_\delta(a)=
 \begin{cases}
 a,&a\notin\ran(p),\\
 p_0(p^{-1}(a)),&a\in\ran(p).
 \end{cases}
\]
Its underlying fiber is exactly $A(a)$, so
$\delta(h_\delta(a))=\operatorname{Copy}_{h_\delta(a)}(0,A(a))$.
Query there and delete the tag and the counter of width
$s(|h_\delta(a)|)$ from every witness answer. Empty fibers give $\bot$
in both directions. Polynomial output bounds for the instance maps and
the composed witness bounds justify writing each query and processing
every legal answer. This proves $\delta\emT A$.

Fix $1\le j\le n$. Completely decode and validate the exact $j$-list,
returning $\bot$ locally for malformed inputs, repetitions, or invalid
witnesses. Return a missing branch constant if one exists. Otherwise, if
a supplied copy has suffix $u$, inspect the counter values $0,\ldots,j$,
written with width $s(|x|)$, and return an absent copy of $u$. These are
$j+1$ distinct available copies because $Q_x\ge n+2>j$; a list of $j$
entries cannot contain them all. The only unresolved valid list has no
copies and contains every constant. A positive list of that form can occur
only on a branch, and its length forces that branch to be $j$ and its set
to be $K_j$. The residual is exactly $\operatorname{Copy}_x(j,B_j(z))$.
Query $B_j(z)$ and encode any witness with tag $j$ and a zero counter;
preserve $\bot$ when $B_j(z)$ is empty. Conversely, given $z$, query
$\delta_{[j]}$ at $p_j(z)$ with all $j$ constants. Delete the tag and
counter, using width $s(|p_j(z)|)$, from every witness answer, and preserve
$\bot$. The query is valid even when $B_j(z)$ is empty. The fixed number
of constants and counter tests and the composed witness bounds give
polynomial-time reductions in both directions.

For $\delta_{[\star]}\lmT B$, let $L$ be the complete source-input
length. Validate the whole positive list, including complete decoding,
membership and distinctness, and return $\bot$ if it is invalid. Let $S$
be its set of $r\ge1$ supplied witnesses. Return a missing constant if
there is one. Group the supplied copies by suffix $u$, and sort the
$s(|x|)$-bit counters within each group. A missing counter can be found
from a missing initial value zero, a gap between consecutive values, or
a missing final value $Q_x-1$. Return the corresponding copy when such
a gap exists. These tests scan only supplied counters. Binary increment
and comparison use $O(s(|x|)+1)$ bits, including the representation of
$Q_x$ as one followed by $s(|x|)$ zeros; they do not enumerate $Q_x$
values to find a gap.

If neither rule succeeds, let $U$ be the set of represented suffixes and
let $Y$ be the appropriate underlying fiber, namely $A(x)$, $A(p(z))$,
or $D_i(z)$. Every represented suffix has all its counter copies in $S$,
and every constant is supplied. The injective copy encoding therefore gives
\[
 S=K\uplus\operatorname{Copy}_x(i,U),\qquad
 \delta(x)\setminus S=\operatorname{Copy}_x(i,Y\setminus U),
\]
where $K=\emptyset$ and $i=0$ on generic and repair inputs.
Suppose first that $U\ne\emptyset$. Then $r\ge Q_x$.
There are $2^{b_*(|x|)+1}-1<Q_x\le r\le L$ candidate suffixes through
the witness-length bound, including $\varepsilon$. Enumerate these
candidates, verify each in the underlying relation, and test exclusion
from $U$. If one succeeds, return its copy with an all-zero counter.
If none succeeds, the displayed residual is empty, so return $\bot$
without a query. This exhaustive search is used only after a nonempty
full grid has been established.

Here is a uniform resource bound for that search and the earlier local
steps. Since $|x|\le L$, all transformed underlying instances have length
at most $\rho(L)$ and all candidates have length at most $b_*(L)$.
One fixed polynomial bounds their verification time, by taking a
majorant for the finitely many verifiers composed with these bounds.
There are fewer than $L$ candidates in the exhaustive case; generation,
comparison with at most $L$ supplied strings, and output writing each
take polynomial time in $L$. List validation, grouping, sorting and binary
counter arithmetic also take polynomial time in $L$, since there are at
most $L$ entries and counter widths are bounded by $s(L)$. Thus one
polynomial in $L$ bounds all these steps, independently of $r$.

If $U=\emptyset$, a positive full grid cannot be generic or repair,
where there are no constants. It is exactly $K_i$ on branch $i$.
For $1\le i\le n$, run the given one-query reduction from $B_i$ to $B$
on $z=p_i^{-1}(x)$. It may halt locally or query $B$ once. Use its decoded
$B_i$ witness, if any, to output a copy with tag $i$ and zero counter,
and return $\bot$ if it returns $\bot$. On branch $n+1$, query $B(z)$
directly and use the same encoding with tag $n+1$. In each case, the
remaining fiber is precisely the copies of $D_i(z)$, so every legal
oracle answer is handled correctly, including empty target fibers. Since
$|z|\le\rho(L)$ and the reduction family is fixed and finite, their
clocks, query lengths, legal $B$-answer lengths and decoded output lengths
have one polynomial bound in $L$. At most one final oracle query is made.
This proves $\delta_{[\star]}\lmT B$.

For $B\lmT\delta_{[\star]}$, given $z$, put $x=p_{n+1}(z)$ and
supply just the $n+1$ constants in $K_{n+1}$. This is a valid positive
list and leaves exactly $\operatorname{Copy}_x(n+1,B(z))$. Its code
has length $2(n+1)+3+2|x|+2(n+1)\ell$, polynomial in $|z|$.
Delete the tag and the instance-dependent counter of width $s(|x|)$
from every witness answer; return $\bot$ on $\bot$. The queried
residual is empty exactly when $B(z)$ is empty, and every witness answer
decodes into $B(z)$. The polynomial bound for $|x|$ and the witness bound
for $\delta$ bound reading and decoding every legal answer. This proves
the last equivalence and completes the construction.
\end{proof}

\begin{corollary}[Finite compatibility of degrees]\label{cor:compatible}
Fix a $1$-paddable language $X$ and an $\FNP$ degree $d_0$ realizable
on $X$. Given $\FNP$ degrees $d_1,\ldots,d_n,d_\star$, there is
$C\in\FNP$ with domain $X$, first degree $d_0$, fixed-level degrees
$d_j$ for $1\le j\le n$, and arbitrary-list degree $d_\star$ if and
only if $d_j\le d_\star$ for every $j$.
\end{corollary}
\begin{proof}
Necessity follows from Lemma~\ref{lem:fixed-star} and composition.
For sufficiency, choose $A$ of degree $d_0$ on $X$ and representatives
$B_j,B$ of the prescribed degrees. Degree domination gives
$B_j\lmT B$, so Theorem~\ref{thm:delta} applies.
\end{proof}

This prescribes only finitely many fixed levels. The extra branch can
impose behavior at unprescribed levels, and no infinite-sequence
prescription follows. The zero-allowed residual task remains different:
Lemma~\ref{lem:res-split} gives $\Res(\delta)\emT A\oplus B$, retaining
the first-solution component. The construction does not require nonempty
targets: if $B$ or a $B_j$ is empty, its constants-only list is exhausted
and the reductions return $\bot$. A $1$-paddable base domain itself is
necessarily infinite, so it cannot be empty.

\section{Comparing fixed-level and arbitrary-list complexity}\label{sec:fixed-variable}

Does the sequence of all fixed-level degrees determine the arbitrary-list degree? Assuming $\PP\ne\NP$, the following construction shows that it does not. We construct a total relation whose first-solution task and every fixed-level additional-solution task are polynomial-time solvable, but whose arbitrary-list task is $\FNP$-complete.
For each instance $x$, we include $\max\{1,|x|\}$ explicitly known witnesses, together with tagged witnesses of an $\FNP$-complete search problem. We call the explicitly known witnesses \emph{dummy witnesses}; they are genuine solutions and ensure that every instance has a solution. A fixed-length list leaves a dummy witness whenever $|x|$ exceeds that length. An arbitrary-length list can contain all dummy witnesses, leaving only the witnesses of the complete search problem.

For fixed $r$, any supplied list of length $r$ leaves an unused dummy witness whenever $|x|>r$. Otherwise, $|x|\le r$, so exhaustive search takes time bounded by a constant depending on $r$. Thus every fixed-level task is polynomial-time solvable. When the list length is part of the input, however, the list can contain all dummy witnesses even for arbitrarily large $x$. Finding another witness then requires solving the complete search problem.

Related padding constructions appear in~\cite[Proposition 2.36]{Str10} and in the $\mathrm{Pad}_t$ construction of~\cite[p.~11]{CS17v2}, which uses dummy outputs and replicated solutions. These constructions concern enumeration. Here, the list of all dummy witnesses has polynomial encoded length, allowing a one-query reduction to arbitrary-list search.

\subsection{A universal efficiently verifiable search relation}

We first construct an $\FNP$-complete relation to use in the counterexample. Its input specifies a verifier, an instance for that verifier, and two bounds: a maximum witness length and a maximum running time. The running-time bound, called a \emph{clock}, stops the simulation after the specified number of steps. Both bounds are encoded in unary: $1^t$ represents $t$ using $t$ bits. Thus the input length accounts for both the permitted witness length and the simulation time.
We write $\operatorname{bin}(i)$ for the binary representation of a positive integer $i$ with no leading zeros.

\begin{lemma}[Clocked universal search]\label{lem:universal-search}
There is a relation $\mathcal U\in\FNP$ that is complete under $\lmT$, and
$\dm(\mathcal U)$ is NP-complete.
\end{lemma}
\begin{proof}
We first define a verifier that works for arbitrary machine descriptions, including descriptions of machines that do not halt. Use deterministic one-tape machines with a two-way infinite tape, alphabet $\{0,1,\mathrm{blank},\#\}$, and states $1,\ldots,k$, where $k\ge3$. State $1$ is the initial state, state $2$ accepts, and state $3$ rejects. The input $\langle a,y\rangle$ is written starting at cell zero, with the head there and all other cells blank.

Encode a description by the unary state count and a list of transitions. Each transition specifies the current state and symbol, the next state, the symbol to write, and the head movement. Use canonical binary state numbers and fixed codes for symbols and movements. Require complete decoding, valid state numbers and codes, and at most one transition for each state--symbol pair. Missing transitions reject. These syntax checks take polynomial time in the description length. For a well-formed input $x=\langle e,a,1^t,1^m\rangle$, define
\[
\mathcal U(x)=\{y:|y|\le m\text{ and the machine described by }e
\text{ accepts }\langle a,y\rangle\text{ within }t\text{ steps}\}.
\]
Malformed tuples, descriptions, or unary fields have empty fibers. An empty unary field represents zero, and a zero clock permits no accepting computation.

To verify a candidate $y$, check the input syntax and reject if $|y|>m$. Otherwise, simulate the described machine for at most $t$ steps and accept exactly when it reaches its accepting state within that bound. Since $t,m\le|x|$, the simulation is polynomially bounded even when the description varies. Indeed, the simulator can scan the transition table at each step and store only the input cells and the at most $t$ additional cells reached during the computation. Both the table and the stored configuration have polynomial length, and there are at most $t$ simulated steps. Thus one fixed polynomial bounds verification time in $|x|+|y|+1$. Every witness has length at most $m\le|x|$, proving $\mathcal U\in\mathrm{FNP}$.

Now fix any $f\in\mathrm{FNP}$. Compile its deterministic graph verifier into the chosen machine model, obtaining a fixed description $e_f$. Let $b_f$ be a monotone polynomial witness bound. Choose a polynomial $t_f$ bounding the compiled verifier's running time on every pair $(a,y)$ with $|y|\le b_f(|a|)$. Such a bound follows by substituting the witness bound into the verifier's polynomial running-time bound.

Define
\[
q_f(a)=\langle e_f,a,1^{t_f(|a|)},1^{b_f(|a|)}\rangle.
\]
The witness cap includes every member of $f(a)$, and the clock allows the verifier to finish on every candidate within that cap. Therefore
\[
\mathcal U(q_f(a))=f(a)
\]
for every $a$, including when $f(a)$ is empty. Query $\mathcal U$ at $q_f(a)$ and return its answer unchanged, including $\bot$. The query is polynomial-time computable and has polynomial length. This proves completeness under the required one-query reductions.

Finally, the domain of $\mathcal U$ belongs to $\mathrm{NP}$: guess a witness of length at most $|x|$ and run the universal verifier. For any language in $\mathrm{NP}$, apply $q_f$ to one of its polynomially bounded witness relations. The displayed fiber equality preserves the existence of a witness, giving a polynomial-time many-one reduction to the domain of $\mathcal U$. Hence that domain is $\mathrm{NP}$-complete.
\end{proof}

\subsection{Many easy witnesses, followed by a complete task}

The next proposition adds dummy witnesses to $\mathcal U$. For each fixed list length $r$, the additional-solution task has a polynomial-time algorithm, but its running-time bound may depend on $r$.

\begin{proposition}[Fixed-easy and variable-complete]\label{prop:fixed-variable}
There is a total polynomial-time-solvable $F\in\FNP$ such that $F_{[r]}$ is
polynomial-time solvable for every fixed $r\ge1$, while $F_{[\star]}$ is
$\FNP$-complete. Moreover, $F_{[\star]}$ is polynomial-time
solvable if and only if $\PP=\NP$.
\end{proposition}
\begin{proof}
Put $d(x)=\max(1,|x|)$ and define
\[
D_x=\{0\operatorname{bin}(i):1\le i\le d(x)\},\qquad
F(x)=D_x\cup\{1y:y\in\mathcal U(x)\}.
\]
The two sets are disjoint because their witnesses begin with different bits. The dummy witness $01$ belongs to every fiber, so $F$ is total and its first-solution task is polynomial-time solvable.

A verifier checks the initial bit. After a zero, it requires a canonical positive binary integer between $1$ and $d(x)$. After a one, it verifies the remaining suffix in $\mathcal U(x)$. These tests take polynomial time. Dummy witnesses have length at most $|x|+2$, and witnesses from the second set have length at most $|x|+1$. Thus $F\in\mathrm{FNP}$ with witness bound $|x|+2$.

Fix $r\ge1$. Decode the input completely, require exactly $r$ supplied witnesses, and check their membership in $F(x)$ and pairwise distinctness. Return $\bot$ if validation fails. If $d(x)>r$, compare the first $r+1$ dummy witnesses with the supplied list. At least one is absent, since the list has only $r$ entries. Return it. If $d(x)\le r$, check all dummy witnesses and return an absent one if possible. Since $|x|\le d(x)\le r$, enumerate every binary string $y$ of length at most $|x|$, including $\varepsilon$, and return $1y$ if $y\in\mathcal U(x)$ and $1y$ is absent from the list. Return $\bot$ if no candidate succeeds.

This checks every possible remaining witness. Validation and list comparisons take polynomial time in the complete input length. The exhaustive search examines at most $2^{r+1}-1$ candidates of length at most $r$, so its additional cost is bounded by a constant depending on the fixed $r$, together with polynomial list processing. Hence $F_{[r]}$ is polynomial-time solvable for every fixed $r$. This bound need not be polynomial when $r$ varies.

Given any input $x$ to $\mathcal U$, form the encoded list containing instance $x$ followed by every witness in $D_x$, ordered by its index. This is always a valid nonempty supplied list for $F$, and its remaining witnesses are exactly
\[
\{1y:y\in\mathcal U(x)\}.
\]
The complete query has length
\[
2d+3+2|x|+2\sum_{i=1}^d(1+|\operatorname{bin}(i)|)
=O((|x|+1)\log(|x|+2)),
\]
where $d=d(x)$, and can be written in polynomial time. Query $F_{[\star]}$ on this list, delete the initial one from any witness answer, and preserve $\bot$. This proves $\mathcal U\lmT F_{[\star]}$, including when $\mathcal U(x)$ is empty.

Lemma~\ref{lem:res-split} gives $F_{[\star]}\in\mathrm{FNP}$. Since $\mathcal U$ is complete by Lemma~\ref{lem:universal-search}, the preceding reduction proves that $F_{[\star]}$ is also complete.

If $F_{[\star]}$ has a polynomial-time refinement, the dummy-list reduction gives one for $\mathcal U$. Testing whether its output is $\bot$ then decides the NP-complete domain of $\mathcal U$ in polynomial time. Consequently, $\mathrm{P}=\mathrm{NP}$.

Conversely, assume $\mathrm{P}=\mathrm{NP}$. For an input $w$, the existence of a witness of $F_{[\star]}(w)$ is an NP question: guess a string within the polynomial witness bound and verify it. The existence of such a witness extending a specified binary prefix is also in NP. Both questions are therefore decidable in polynomial time.

First test whether a witness exists, returning $\bot$ if none does. Otherwise, start with the empty prefix. If the current prefix is itself a witness, return it. If not, append zero if the current prefix followed by zero extends to a witness, and append one otherwise. At every stage the chosen prefix extends to a witness, and its length never exceeds the polynomial witness bound. This constructs a witness in polynomial time and proves the converse.
\end{proof}

\begin{corollary}[Conditional failure of determination by fixed levels]
\label{cor:spectrum-star}
Assume $\PP\ne\NP$. Let $F$ be the relation from the preceding proposition and define $\mathbf1(x)=\{\varepsilon\}$ for every $x$. Both relations have domain $\Sig^*$ and polynomial-time-solvable first-solution tasks. For every fixed $r\ge1$,
\[
\dg F_{[r]}=\dg\mathbf1_{[r]},
\]
but
\[
\dg F_{[\star]}\ne\dg\mathbf1_{[\star]}.
\]
Thus the domain, first-solution degree, and entire sequence of fixed-level degrees do not determine the arbitrary-list degree.
\end{corollary}
\begin{proof}
Both relations have a witness on every instance, and their first-solution tasks are polynomial-time solvable: return $\varepsilon$ for $\mathbf1$ and $01$ for $F$.

For $\mathbf1$, a valid one-witness list contains $\varepsilon$ and is exhausted. A longer list cannot be valid because supplied witnesses must be distinct. Thus every $\mathbf1_{[r]}$ and $\mathbf1_{[\star]}$ is everywhere empty and is solved by always returning $\bot$. The preceding proposition shows that every fixed-level task $F_{[r]}$ is also polynomial-time solvable. Hence the two relations have the same domain, first-solution degree, and degree at every fixed positive level. However, under $\mathrm{P}\ne\mathrm{NP}$, the proposition shows that $F_{[\star]}$ is not polynomial-time solvable. Its degree therefore differs from that of $\mathbf1_{[\star]}$.
\end{proof}

\subsection{What a global constant cap changes}

The previous example uses more dummies on longer inputs. A fixed upper
bound on every fiber prevents this mechanism: only finitely many supplied
arities can leave an answer. The cap must be a constant shared by all
instances, not merely a polynomial bound depending on the instance.

\begin{proposition}[Constant fiber caps]\label{prop:constant-cap}
Let $f\in\FNP$ and let $M\ge0$ be fixed with $|f(x)|\le M$ for all $x$.
With raw $f_{[0]}=f$,
\[
 f_{[\star]}\emT\bigoplus_{1\le r<M}f_{[r]},\qquad
 \Res(f)\emT\bigoplus_{0\le r<M}f_{[r]}.
\]
Consequently, easy fixed positive levels imply easy variable-list search
under a constant cap. Two constant-capped relations with the same fixed
positive degree sequence have the same star degree; equality also at raw
level zero gives equality of their $\Res$ degrees.
\end{proposition}
\begin{proof}
If $M=0$, every fiber of $f$ is empty. Both residual operators and both displayed sums are therefore solved by always returning $\bot$. Assume $M\ge1$.

To reduce a variable-list operator to the finite sum, decode and validate the source list, requiring $r\ge1$ for $f_{[\star]}$ and allowing $r=0$ for $\operatorname{Res}(f)$. Return $\bot$ if validation fails. If $r\ge M$, also return $\bot$: a valid list of length $M$ contains every witness, while a longer list cannot be valid.

If $r<M$, query the component indexed by $r$. For $r\ge1$, use the original supplied-list code as its inner input. For $r=0$, use the raw instance $x$, since $f_{[0]}=f$. Delete the component selector from any witness answer and preserve $\bot$. The selected component has exactly the source task's remaining witnesses.

For the reverse reduction, decode the outer input completely and check that its selector names a component in the sum. Return $\bot$ if this fails. For a positive component $r$, validate the inner input as a list containing exactly $r$ distinct witnesses of $f(x)$; return $\bot$ if it is invalid. Otherwise, query the corresponding variable-list operator on that same list. For component zero, the inner input is a raw instance $x$; query $\operatorname{Res}(f)$ at $\langle x\rangle$.

Wrap any witness answer in the original component selector and preserve $\bot$. This gives precisely the answers required by the source component, including $\bot$ on exhausted lists. All list processing is polynomial-time, and the fixed finite family gives one polynomial bound for query writing and answer processing. Each reduction uses at most one query.

For two relations with different constant caps, enlarge both sums to the larger cap. The added components are everywhere empty, since a list at or above the original cap cannot leave another witness. If the fixed positive degrees agree, the corresponding components are equivalent, so the star operators are equivalent. Agreement also at level zero gives equivalence of the residual operators.

When $M=1$, the positive sum is empty and $\operatorname{Res}(f)\emT f$. The argument requires a constant cap shared by all instances: an instance-dependent bound does not reduce the task to a fixed finite family of levels.
\end{proof}

\section{Totality, limitations, and remaining questions}\label{sec:limits}

\subsection{When does totalization produce total search?}\label{sec:totality}

A total search relation has a witness on every input. By contrast, invalid-input totalization adds the witness $\varepsilon$ only to invalid supplied-list inputs. A valid list containing every available witness still has no answer. Thus the totalized operator is total exactly when no valid list of the permitted length exhausts a fiber. This condition depends only on fiber sizes, with no complexity assumption.

\begin{observation}[Exact totality criterion]\label{obs:totality}
For $f\in\FNP$ and fixed $r\ge1$,
\[
 f_{[r]t}\in\TFNP\quad\Longleftrightarrow\quad
 \forall x\ |f(x)|\ne r.
\]
Also $f_{[\star]t}$ is total exactly when $f$ is empty, and no nonempty
$f$ has every positive $f_{[r]t}$ total.
\end{observation}
\begin{proof}
Invalid inputs always receive the witness $\varepsilon$, so totality depends only on valid supplied lists. If $|f(x)|=r$, listing all $r$ witnesses gives a valid exhausted input, on which $f_{[r]t}$ has no answer. Thus the totalized operator is not total.

Conversely, suppose no fiber has size exactly $r$. A valid list of $r$ distinct witnesses implies $|f(x)|\ge r$, and hence $|f(x)|>r$. At least one witness remains outside the list. Therefore every valid input has an answer, and invalid inputs have $\varepsilon$, proving totality.

If $f$ is empty, every positive supplied list is invalid, so $f_{[\star]t}$ is total. If $f$ is nonempty, choose a nonempty fiber and let its size be $m$. The fiber is finite, so listing all its witnesses gives a valid exhausted input for both $f_{[\star]t}$ and $f_{[m]t}$. Neither is total. Consequently, no nonempty relation has all its positive fixed-level totalizations total.

The exhausted list need only exist; this argument does not require a polynomial-time procedure for constructing it.
\end{proof}

\subsection{The total-search reduction convention}\label{sec:total-search}

For total search, one usually describes a reduction by an instance map
and an answer decoder. Because every target query has a witness, this
convention agrees with our one-query model in the total setting. We spell
out the bridge to avoid using it for partial residuals.
For $f,g\in\TFNP$, write $f\sm g$ when total polynomial-time maps
$a(x)$ and $\operatorname{dec}(x,y)$ satisfy
$\operatorname{dec}(x,y)\in f(x)$ for every $y\in g(a(x))$.
These maps give a one-query reduction: query $g(a(x))$ and decode its
answer. If $|a(x)|\le q(|x|)$, then every legal answer has length at most
$b_g(q(|x|))$. The decoder's polynomial bound in $|x|+|y|+1$ is thus
polynomial in $|x|$, and its legal output is also bounded by $b_f(|x|)$.

Conversely, clock a one-query reduction by its polynomial running-time
bound. Run it to its query or no-query halt, and use that query as $a(x)$;
in a no-query case use the fixed target instance $\varepsilon$. The decoder
repeats the clocked computation and substitutes $y$ for its oracle answer,
or ignores $y$ after a no-query halt. Define it on all arguments by returning
a fixed binary default on an illegal answer, timeout, second query, or
invalid output. Illegal answers can be detected using $g$'s graph verifier,
and outputs using $f$'s; $\bot$ is illegal for a total target. The clock
bounds the query and simulation, while the $\FNP$ witness bounds bound
all legal answers and outputs. On legal arguments the guards never apply,
since any legal answer extends to a refinement of $g$. Thus $\sm$ and
$\lmT$ agree for $f,g\in\TFNP$, including no-query cases. For partial
targets, the instance-map condition alone does not specify how an oracle
answer $\bot$ must be handled.

\subsection{A total construction for one fixed level}

For a chosen fixed level, totality is easy to guarantee: include that many
constants and at least one target witness in every fiber. Every valid list
of the chosen size then leaves an answer. This is a single-level claim;
Observation~\ref{obs:totality} excludes simultaneous totality at all
positive levels for a nonempty finite-fiber relation.

\begin{remark}[Realizing a total target at a fixed level]\label{rem:total-realization}
Fix $r\ge1$ and $B\in\TFNP$. There is a total polynomial-time-solvable relation $H_B\in\FNP$ such that
\[
(H_B)_{[r]t}\in\TFNP,\qquad
(H_B)_{[r]t}\equiv^p_{\rm sm}B,
\qquad
(H_B)_{[r]t}\emT B.
\]
Thus finding a first witness of $H_B$ is easy, while finding another after $r$ supplied witnesses has exactly the degree of $B$. For $r=1$, this is analogous to the construction and conclusions of~\cite[Definition~4.5, Theorem~4.6, and Corollary~4.9]{EGG25}: the theorem gives search equivalence and the corollary concerns totality. The self-contained argument below works for every fixed $r$. The non-totalized operator $(H_B)_{[r]}$ is also one-query equivalent to $B$, but is empty on invalid inputs and hence is not total.

Use the shared tag set $\{0,\ldots,r\}$ with width $\ell=\lceil\log_2(r+1)\rceil$. Let $K=\{\bl t:0\le t<r\}$ and define
\[
H_B(x)=K\uplus\{\bl r y:y\in B(x)\}.
\]
Each instance has $r$ explicit constants and at least one tagged $B$-witness. A constant therefore solves the first-solution task, and every valid list of $r$ witnesses leaves another witness. Invalid inputs receive $\varepsilon$ under totalization, so $(H_B)_{[r]t}$ is total. Its graph is polynomial-time decidable, and $\ell+b_B(|x|)$ bounds its witness lengths.

To reduce the additional-solution task to $B$, first validate the supplied list. On an invalid input, return $\varepsilon$. On a valid input missing a constant, return that constant. Otherwise, the $r$ supplied witnesses are exactly $K$. Query $B(x)$ and return $\bl r y$ for any answer $y$. To express this as an instance map and decoder, use the fixed target instance $\varepsilon$ in the two locally solved cases and let the decoder ignore the oracle answer. In the remaining case, use target instance $x$ and prefix the answer with $\bl r$.

For the reverse reduction, given $x$, supply all constants in $K$. This list is valid and leaves exactly
\[
\{\bl r y:y\in B(x)\}.
\]
Delete the tag from any answer to obtain a witness of $B(x)$. The query has length $2r+3+2|x|+2r\ell$. Since $r$ is fixed, both reductions have polynomial-time instance maps and decoders. Their equivalence under $\emT$ follows from the total-search reduction convention. For $(H_B)_{[r]}$, use the same two reductions, returning $\bot$ locally on invalid source lists. The reverse query still supplies the valid list $K$, so its legal witness answers still decode into $B(x)$. This proves its one-query equivalence separately from totalization.
\end{remark}

\subsection{Remaining questions}\label{sec:questions}

Our reductions validate supplied witnesses in polynomial time before making an oracle query. We have not extended the results to search classes in which verifying a witness is only known to be in $\NP$.

The fixed-domain constructions require $1$-paddability. Some restriction on the domain is necessary: on a finite domain, every $\FNP$ relation has only finitely many instance--witness pairs and can be solved by a fixed table. However, $1$-paddability also excludes sparse domains. Indeed, if $|p(z)|\le b(|z|)$ for a monotone polynomial $b$, injectivity gives at least $2^{a+1}-1$ domain strings of length at most $b(a)$. Can simultaneous realization hold on infinite domains under weaker assumptions? The corresponding single-level question appears in~\cite[Section~5]{EGG25}.

Theorem~\ref{thm:delta} settles finite compatibility: the arbitrary positive-list degree can be prescribed exactly when it dominates all the finitely many prescribed fixed-level degrees. Its extra branch also determines some unprescribed behavior, so the theorem does not settle prescribing an entire infinite fixed-level sequence together with an arbitrary-list degree. We also do not address retaining unchanged natural solution sets or computing all remaining witnesses, the third direction raised in~\cite[Section~5]{EGG25}.

\section{Conclusion}

Under $1$-paddability, we can prescribe any finite collection of fixed additional-solution degrees while preserving the domain and first-solution degree. The arbitrary positive-list degree can be prescribed as well if and only if it dominates every prescribed fixed-level degree. The separate counter-copy construction proves sufficiency without a P-versus-NP assumption.

For the finite construction, we characterize arbitrary-list search by a one-query equivalence with the component residual tasks. We also construct a total relation whose first-solution task and every fixed additional-solution level are polynomial-time solvable, while its arbitrary-list task is $\FNP$-complete. Assuming $\PP\ne\NP$, even the domain, first-solution degree, and entire fixed-level degree sequence therefore do not determine the arbitrary-list degree.

These constructions define new witness sets. They do not establish the same prescription results when a problem's original notion of solution must remain unchanged.

\end{document}